\documentclass[a4,12pt]{amsart}

\usepackage{amssymb,amsthm,amsmath}
\usepackage{graphicx}
\usepackage{hyperref}
\usepackage{color}
\usepackage{enumerate}
 \usepackage{txfonts}
 \usepackage{tikz}
\usepackage[margin=2cm]{geometry}

\theoremstyle{definition}
\newtheorem{theorem}{Theorem}[section]

\newtheorem{problem}[theorem]{Problem}
\newtheorem{proposition}[theorem]{Proposition}
\newtheorem{definition}[theorem]{Definition}
\newtheorem{remark}[theorem]{\rm Remark}
\newtheorem{example}[theorem]{Example}
\numberwithin{equation}{section}

\newcommand{\Z}{\mathbb{Z}}

\usepackage[boxruled,norelsize]{algorithm2e5}

\title{Dappled tiling} 

\author{Shizuo Kaji}
\thanks{The first author was partially supported by KAKENHI, Grant-in-Aid for Young 
     Scientists (B) 26800043.}
\address{Yamaguchi University, Japan / JST PRESTO}
\email{skaji@yamaguchi-u.ac.jp}

\author{Alexandre Derouet-Jourdan}
\address{OLM Digital Inc.}
\email{alex@olm.co.jp}

\author{Hiroyuki Ochiai}
\address{Kyushu University}
\email{ochiai@imi.kyushu-u.ac.jp}
\thanks{The third author was partially supported by KAKENHI Grant Number 15H03613.}


\begin{document}

\maketitle

\begin{abstract}
We consider a certain tiling problem of a planar region
 in which there are no long horizontal or vertical strips consisting of 
 copies of the same tile. Intuitively speaking, we 
 would like to create a dappled pattern with 
 two or more kinds of tiles.
We give an efficient algorithm to turn any tiling into one satisfying the condition,
and discuss its applications in texturing.
\end{abstract}

\section{Introduction}
In texturing, we often encounter the following problem:
fill a region with a given collection of small square patches in such a way that
patches of a same kind do not appear in a row.
We make this problem more precise.

\begin{definition}
For natural numbers $m$ and $n$, let $G_{m,n}$ be a rectangular grid
\[
G_{m,n}=\{ (i,j)\in \Z\times \Z \mid 0\le i < m, 0\le j< n \}.
\]
We call its elements {\em cells}.
Our convention is that $(0,0)$ is the cell at the top-left corner and $(m-1,0)$ is at the top-right corner.
For a finite set of tiles $T$ with $\# T\ge 2$, 
we call a function $G_{m,n}\to T$ a {\em tiling} of $G_{m,n}$ with $T$.
For a natural number $p>1$ and $t\in T$, we say $f$ satisfies the condition $H^p_t$
if there is no horizontal strip with more than $p$ consecutive $t$'s, that is, there is no $(i,j)\in G_{m,n}$
such that $f(i,j)=f(i-1,j)=\cdots=f(i-p,j)=t$.
Similarly, we say $f$ satisfies the condition $V^q_t$  for a natural number $q>1$,
if there is no vertical strip with more than $q$ consecutive $t$'s.
\end{definition}

Consider a set $L$ consisting of conditions of the form
$H^{p}_t$ and $V^{q}_t$ with varying $p,q>1$ and $t\in T$.
Alternatively, 
we can think of $p,q$ as functions
$p,q: T \to \{2,3,4,\ldots,\}\cup \{\infty\}$
so that $L=\bigcup_{t\in T} \{H^{p(t)}_t, V^{q(t)}_t\}$.
For notational convenience, we allow $H^\infty_t$, which is always satisfied.
We will use both notations interchangeably.
We say a tiling $f$ is $L$-\emph{dappled} if $f$ satisfies all the conditions in $L$.
The problem we are concerned is:
\begin{problem}
Give an efficient algorithm to produce
$L$-dappled tilings, 
which posses some controllability by the user.
\end{problem}
In this paper, we introduce an algorithm to produce an $L$-dappled tiling by 
modifying a given initial tiling which need not be $L$-dappled.
Note that enumerating all the $L$-dappled tilings is fairly straightforward;
we can fill cells sequentially from the top-left corner.
However, this is not practical since there are exponentially many $L$-dappled tilings with respect to 
the number of cells,
and many of them are not suitable for applications
as we see below.
\begin{proposition}
Let $N=\lceil \frac{m}{2}\rceil \lceil\frac{n}{2}\rceil$.
There exist at least $|T|^N$ tilings which are $L$-dappled.
\end{proposition}
\begin{proof}
We will create {\em draughtboard} tilings.
For each cell $(2k,2l)$,
choose any tile $t\in T$ and put the same tile at $(2k+1,2l+1)$ (if it exists).
Pick any $t',t''\in T\setminus \{t\}$ and put them at $(2k+1,2l)$ and $(2k,2l+1)$ (if they exist).
One can see that for any $(i,j)\in G_{m,n}$ the tile at $(i-1,j)$ or $(i-2,j)$
is different from the one at $(i,j)$. Similarly,
the tile of $(i,j-1)$ or $(i,j-2)$ is different from the one at $(i,j)$,
and hence, the tiling thus obtained is $L$-dappled with any $L$.
There are $N$ cells of the form $(2k,2l)$, and hence,
there are at least $|T|^N$ draughtboard tilings.
\end{proof}
It is easy to see that the above argument actually shows that 
there are at least $\big(|T|(|T|-1)^2+|T|(|T|-1)(|T|-2)^2 \big)^{N'}$
draughtboard (and hence, $L$-dappled)
tilings with $N'=\lfloor \frac{m}{2}\rfloor \lfloor\frac{n}{2}\rfloor$.
\begin{example}
We show an example of a draughtboard tiling with $T=\{\text{orange, white}\}$ (Fig. \ref{fig:draughtboard}).
For any set of conditions $L$, it is an $L$-dappled tiling.
\begin{figure}[htb]
    \centering
    \includegraphics[width=0.05\linewidth]{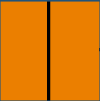} \qquad
    \includegraphics[width=0.05\linewidth]{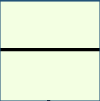} \qquad 
    \includegraphics[width=0.4\linewidth]{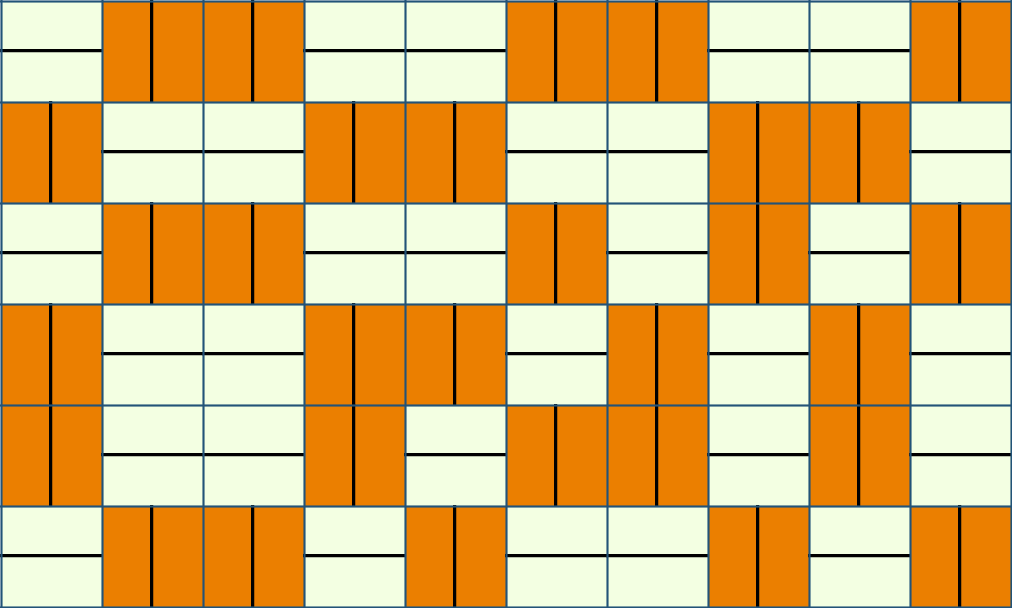} \qquad
    \caption{Orange and White tiles and an example of a draughtboard tiling of $G_{10,6}$}
    \label{fig:draughtboard}
\end{figure}\end{example}
Draughtboard patters look very artificial 
and are not suitable for texturing.
We would like to have something more natural.
Therefore, we consider an algorithm to produce 
an $L$-dappled tiling which takes a user specified (not necessarily $L$-dappled) tiling as input
so that the user has some control over the output.
We also discuss a concrete applications with the Brick Wang tiles (\cite{meis2015,AKM,scss2016}) in \S \ref{sec:brick}, 
and with flow generation in \S \ref{sec:flow}.

\begin{remark}
For the special case of $T=\{0,1\}$ and $\{H^2_0,V^2_1\}$,
the numbers of $L$-dappled tilings
for several small $m$ and $n$ are listed at \cite{A206994}.
No formula for general $m$ and $n$ nor a generating function 
is known as far as the authors are aware.
\end{remark}

\section{The algorithm}
Fix a set of conditions $L$. We just say dappled for $L$-dappled from now on.
Given any tiling $f$, we give an algorithm to convert it into a dappled one.
We can start with a random tiling or a user specified one.
First, note that the problem becomes trivial when $|T|>2$
since we can choose a tile for $f(i,j)$ at step (I) below 
which is different from $f(i-1,i)$ and $f(i,j-1)$.
So, we assume $T$ consists of two elements $\{0, 1\}$.

The idea is to perform ``local surgery'' on $f$.
We say $f$ {\em violates} the condition $H^p_t\in L$ (resp. $V^q_t\in L$) at $(i,j)$
when $f(i,j)=f(i-1,j)=\cdots=f(i-p,j)=t$
(resp. $f(i,j)=f(i,j-1)=\cdots=f(i,j-q)=t$).
For a cell $(i,j)$ we define its {\em weight} $|(i,j)|=i+j$.
Let $(i,j)$ be a cell with the minimum weight
 such that $f$ violates any of the conditions $H^p_t\in L$ or $V^q_t\in L$.
We modify values of $f$ around $(i,j)$ to rectify the violation in the following manner.
\begin{enumerate}
\renewcommand{\labelenumi}{(\Roman{enumi}).}
\item Set $f(i,j)=1-t$ if it does not violate any condition at $(i,j)$ in $L$.
\item Otherwise, set $f(i,j)=f(i-1,j-1), f(i-1,j)=1-f(i-1,j-1)$, and $f(i,j-1)=1-f(i-1,j-1)$.
\end{enumerate}
Let us take a close look at step (II).
Assume that $f$ violated $H^p_t$ at $(i,j)$. This means $f(i-2,j)=f(i-1,j)=f(i,j)=t$.
Note also that $f(i,j-1)=f(i,j-2)=1-t$ since otherwise we could set $f(i,j)=1-t$ at step (I).
When $f(i-1,j-1)=t$, we can set $f(i-1,j)=1-t$ without introducing a new violation at $(i-1,j)$.
When $f(i-1,j-1)=1-t$, we can set $f(i,j)=1-t$ and $f(i,j-1)=t$ without introducing a new violation at either of $(i-1,j)$ or $(i,j-1)$.
A similar argument also holds when $V^q_t$ is violated at $(i,j)$.

After the above procedure,
the violation at $(i,j)$ is resolved without introducing a new violation at
cells with weight $\le i+j$.
(We successfully ``pushed'' the violation forward.)
Notice that each time either the minimal weight of violating cells increases or
the number of violating cells with the minimal weight decreases.
Therefore, by repeating this procedure a finite number of times, 
we are guaranteed to obtain a dappled tiling transformed from the initially given tiling.

The algorithm works in whatever order the cells of a same weight are visited,
but our convention in this paper is in increasing order of $i$.
All the examples are produced using this ordering.

\begin{proposition}
Fix any $m,n>0$, $T=\{0,1\}$, and
$L=\{H^{p(t)}_t, V^{q(t)}_t\mid t\in T\}$
with $p(t), q(t)>1$ for all $t\in T$.
Algorithm \ref{algorithm} takes a tiling
$f: G_{m,n}\to T$ and outputs an $L$-dappled tiling.
If $f$ is already $L$-dappled,
the algorithm outputs $f$.
In other words, Algorithm \ref{algorithm} defines a 
retraction from the set of tilings of $G_{m,n}$ onto that of $L$-dappled tilings of $G_{m,n}$.
\end{proposition}

\begin{algorithm}[H]
\SetKw{KwCont}{continue}
\SetKw{KwBreak}{break}
\SetKw{KwAnd}{and}
\SetKw{KwNot}{not}
\SetKwFunction{TestLeft}{Violate}
\SetKwProg{Fn}{Function}{}{}
\SetAlgoLined
 \KwIn{A tiling $f:G_{m,n}\to T$, a set of conditions $L$}
 \KwOut{An $L$-dappled tiling $g: G_{m,n}\to T$}
 (note that in the below the values of $f$ and $g$ for negative indices should be understood appropriately)
 
\Begin{
 $g \gets f$ \;
 \For{$weight=0$ \KwTo $m+n-2$}{
	 \ForAll{$(i,j)\in G_{m,n}$ such that $i+j= weight$}{
	 	\If{\TestLeft$(g,(i,j))$}{
		$g(i,j) \gets 1-g(i,j)$\;
		\If{\TestLeft$(g,(i,j))$}{
	 		$g(i,j) \gets g(i-1,j-1)$ \;
       		$g(i-1,j) \gets 1-g(i,j)$ \;
       		$g(i,j-1) \gets 1-g(i,j)$ \;
		}}
    }
 }
 \Return{$g$}
}
\Fn{\TestLeft$(f, (i,j))$}{
\ForAll{$H^p_t \in L$}{
	\If{$f(i,j)=f(i-1,j)=\cdots=f(i-p,j)=t$}{\Return{true}}
}
\ForAll{$V^q_t \in L$}{
	\If{$f(i,j)=f(i,j-1)=\cdots=f(i,j-q)=t$}{\Return{true}}
}
\Return{false}
}
\caption{Algorithm to convert an input tiling to an $L$-dappled one.}
\label{algorithm}
\end{algorithm}
The sub-routine $\mathrm{Violate}$ returns true if $f$ violates any of horizontal or vertical 
conditions at the given cell.
In practice, the check can be efficiently performed by book-keeping the numbers of consecutive tiles 
of smaller weight in the horizontal and the vertical directions.
See the python implementation \cite{code} for details.
\begin{example}
Fig. \ref{fig:alg} shows how Algorithm \ref{algorithm} proceeds for
$T=\{\text{white, orange}\}$ and $L=\{H^{2}_{\text{white}}, V^{2}_{\text{orange}}\}$.
\begin{figure}[htb]
    \centering
    \includegraphics[width=0.3\linewidth]{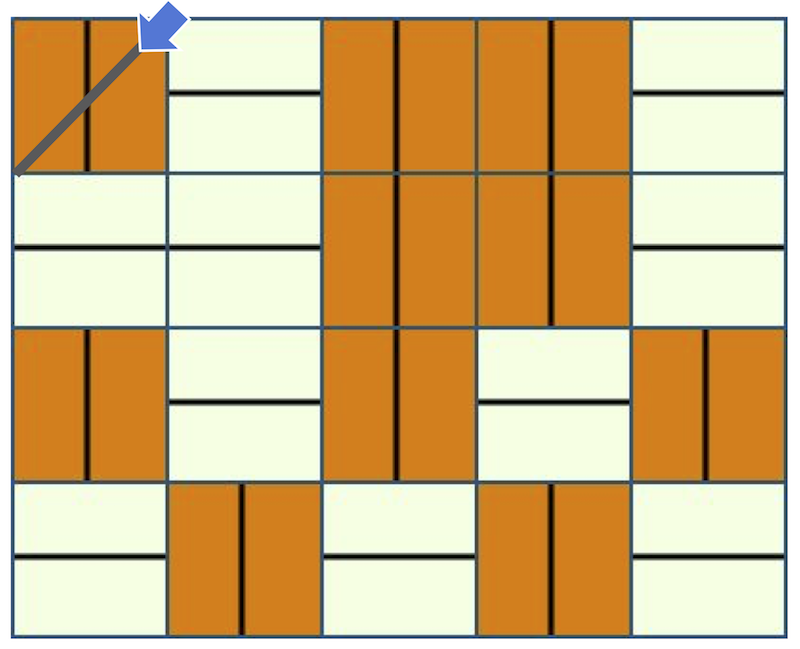}
    \includegraphics[width=0.3\linewidth]{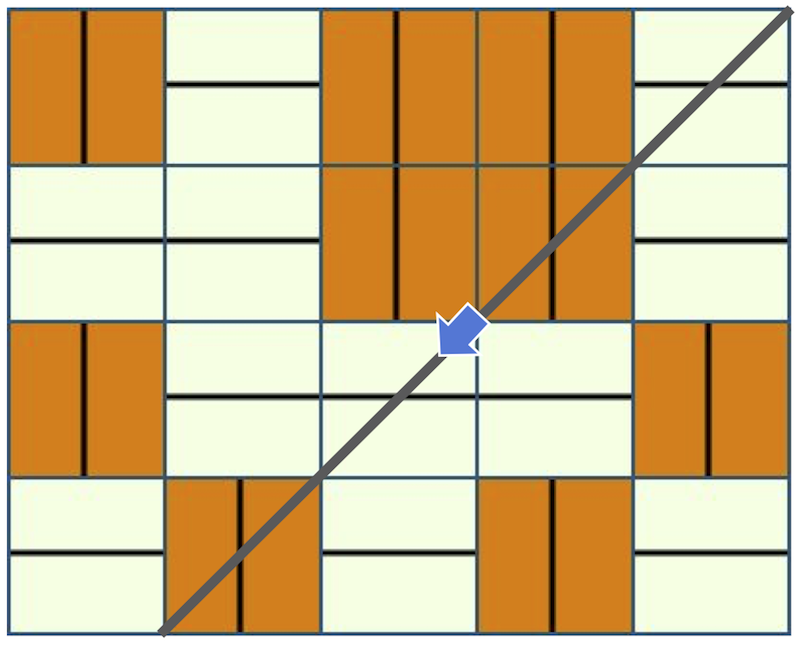}
    \includegraphics[width=0.3\linewidth]{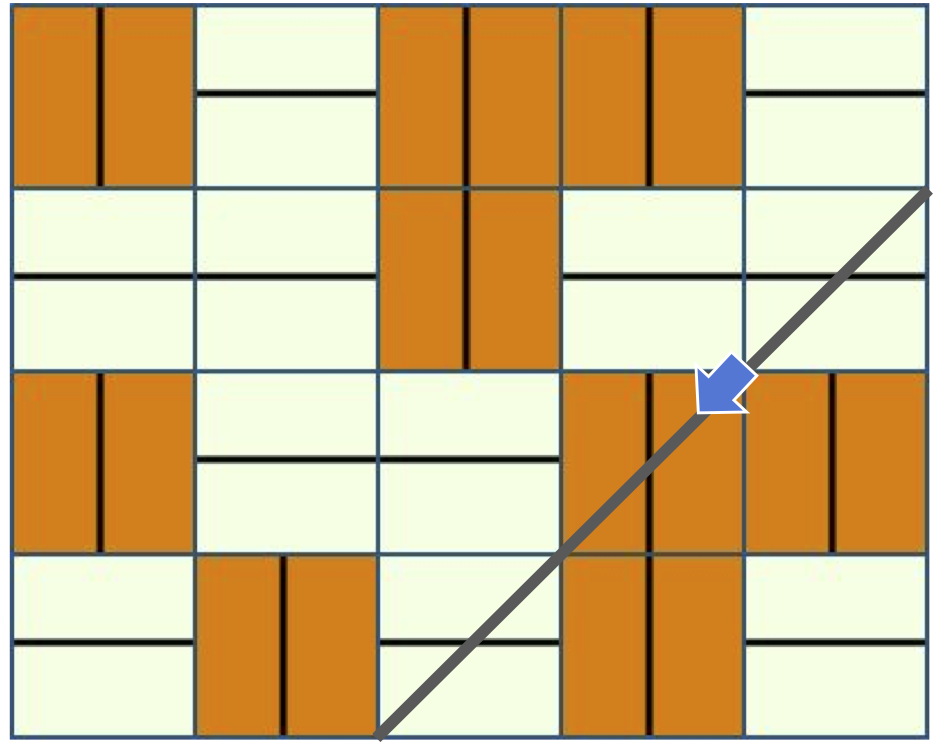}
    \caption{Steps of Algorithm \ref{algorithm}. Left: input tiling, Middle: resolving the violation of $V^2_{\text{orange}}$ at $(2,2)$ by (I),
    Right: resolving the violation of $H^2_{\text{white}}$ at $(3,2)$ by (II).}
    \label{fig:alg}
\end{figure}
\end{example}

\begin{remark}\label{rem:p=1}
Algorithm \ref{algorithm} does not always work when 
$p(t)=1$ or $q(t)=1$ for some $t\in T$.
For example, when $L=\{H^1_0, H^2_1, V^2_1\}$
it cannot rectify the following tiling of $G_{4,3}$:
\[
\begin{matrix}
1 & 0 & 1 & 1 \\
0 & 1 & 0 & 1 \\
1 & 1 & 0 & 0
\end{matrix}
\]
\end{remark}

\section{Extension}
We give two extensions of the main algorithm discussed in the previous section.
\subsection{Non-uniform condition}
It is easy to see that our algorithm works when 
the conditions $H^p_t$ and $V^q_t$
vary over cells.
That is, $p$ and $q$ can be functions of $(i,j)\in G_{m,n}$
as well as $t\in T$ so that
$p,q: T\times G_{m,n} \to \{2,3,4,\ldots,\}\cup \{\infty\}$.
This allows the user more control over the output.
For example, 
the user can put non-uniform constraints,
or even dynamically assign constraints computed from the 
initial tiling.

\begin{example}
Let $T=\{\text{white,orange}\}$ and 
\[
L=\{H^{p(\text{white};i,j)}_{\text{white}}, H^{p(\text{orange};i,j)}_{\text{orange}}, 
V^{q(\text{white};i,j)}_{\text{white}}, V^{q(\text{orange};i,j)}_{\text{orange}} \},
\]
where $p(\text{white};i,j)=q(\text{orange};i,j)=\lceil \frac{i+1}{5} \rceil+1$ and
$p(\text{orange};i,j)=q(\text{white};i,j)=\lceil \frac{m-i}{5} \rceil+1$.
An example of $L$-dappled tiling is given in Fig. \ref{fig:non_uni_cons}.
In the left area, 
long horizontal white strips 
and long vertical orange strips are prohibited,
while in the right area, 
long vertical white strips and
long horizontal orange strips are prohibited.
\begin{figure}[ht]
  \centering
      \includegraphics[width=0.5\linewidth]{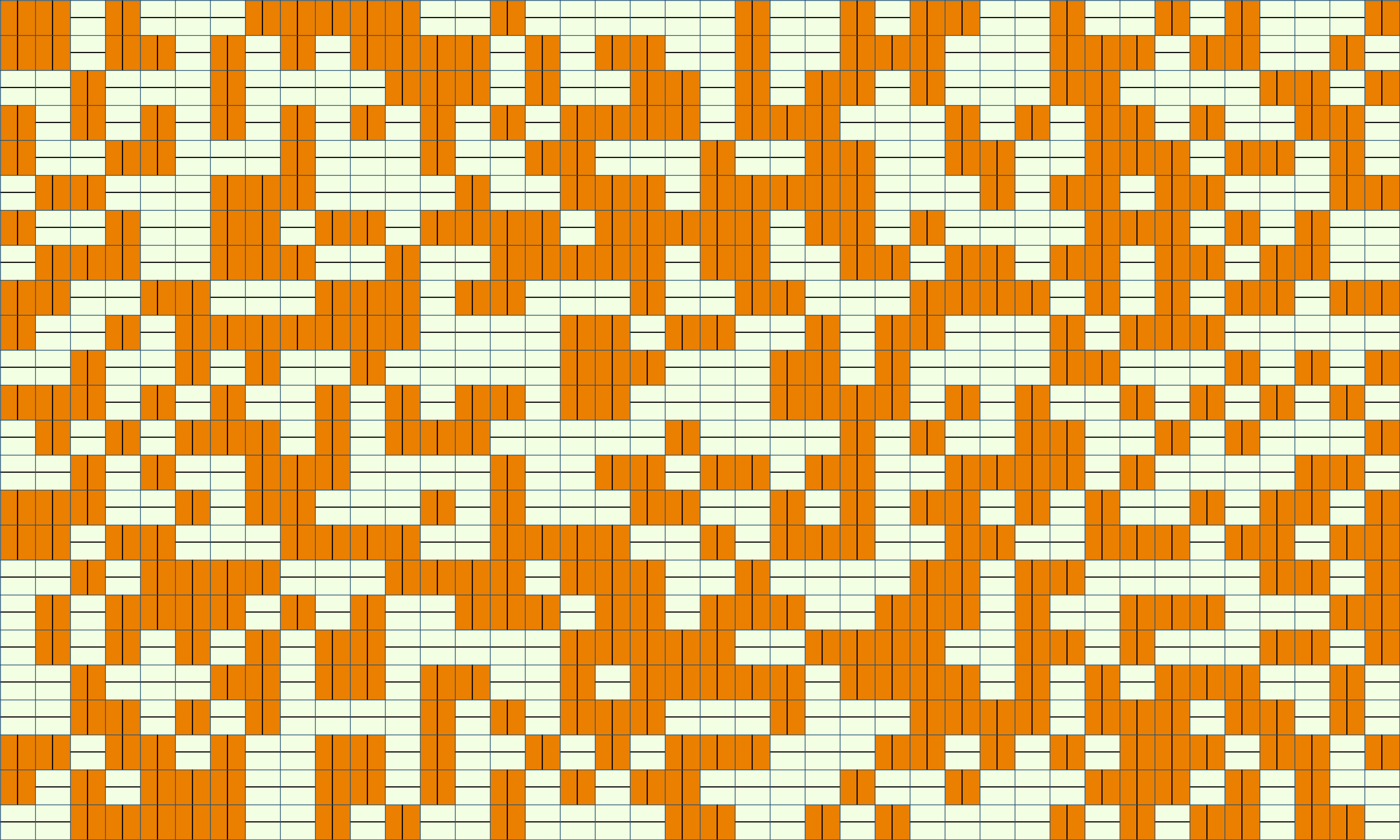}
      \caption{Dappled tiling with non uniform conditions.} \label{fig:non_uni_cons}
\end{figure}
\end{example}

\subsection{Cyclic tiling}\label{sec:cyclic}
Sometimes we would like to have an $L$-dappled tiling of $G_{m,n}$
which can be repeated to fill a larger region, say $G_{km,ln}$ for $k,l\ge 1$.
For this, the conditions have to be understood as being {\em cyclic}; 
for example, $\bar{H}^p_t$ is violated if there is a cell $(i,j)$
with $f(i,j)=f([i-1],j)=\cdots f([i-p],j)=t$,
where $0\le [x] \le m-1$ is the reminder of $x$ divided by $m$.
For a set $\bar{L}$ consisting of conditions of the form
$\bar{H}^{p}_t$ and $\bar{V}^{q}_t$,
we say a tiling $f$ is {\em cyclically $\bar{L}$-dappled}
if it does not violate any of the conditions in $\bar{L}$
in the above cyclic sense.

We discuss a modification of Algorithm \ref{algorithm} to produce a cyclically $\bar{L}$-dappled tiling.
However, there are two limitations: it only works for a limited class
of conditions; when $T=\{0,1\}$, 
we have to assume $\bar{L}$ should satisfy 
$p(t),q(t)>2$ for all $t\in T$
(see Example \ref{cyclic:fail}).
The other drawback is that the algorithm 
changes an input tiling even when it is already cyclically $\bar{L}$-dappled.
This is because it produces a cyclically $\bar{L}$-dappled tiling with 
additional conditions.

Let $f$ be any tiling. 
We introduce Algorithm \ref{cyc-algorithm}, which is a modification of Algorithm \ref{algorithm}.
We visit cells in increasing order of the weight as in Algorithm \ref{algorithm}.
When the cell $(i,j)$ is visited,
we define a set of non-cyclic conditions $L(i,j)$ which is more stringent than $\bar{L}$.
For each $\bar{H}^p_t\in \bar{L}$,
\begin{enumerate}
\renewcommand{\labelenumi}{(\roman{enumi}).}
\item skip if $i<p-2$
\item add $H^{p-2}_t$ to $L(i,j)$ if $i=p-2$
\item add $H^{p-k}_t$ to $L(i,j)$ if $i=m-1$,
where $k$ is the smallest non-negative integer such that $f(k, j) \neq t$.
\item add $H^p_t$ to $L(i,j)$ otherwise.
\end{enumerate}
And do similarly for $\bar{V}^q_t$.
Then, resolve (if any) violation of $L(i,j)$ at $(i,j)$
 in the non-cyclic sense using Algorithm \ref{algorithm}.
By (ii) it is ensured that there exists $k\le p-2$ such that $f(k,j)\neq t$ if $\bar{H}^p_t\in \bar{L}$.
Note that although we have to impose $H^1_t$ at $(1,j)$ when $p=3$,
Algorithm \ref{algorithm} works with no problem in this case.
For (iii) we always have $p-k\ge 2$ since $(m-1,j)$ is visited later than $(p-2,j)$,
and $k$ must be less than or equal to $p-2$ by (ii).

Due to the extra condition imposed by (ii),
the output is in a restricted class of cyclically $\bar{L}$-dappled tilings.

\begin{proposition}
Fix any $m,n>0$, $T=\{0,1\}$, and
$\bar{L}=\{\bar{H}^{p(t)}_t, \bar{V}^{q(t)}_t\mid t\in T\}$
with $p(t),q(t)>2$ for all $t\in T$.
Algorithm \ref{cyc-algorithm} takes a tiling
$f: G_{m,n}\to T$ and outputs a cyclically $\bar{L}$-dappled tiling.
\end{proposition}

\begin{algorithm}[htb]
\SetKw{KwCont}{continue}
\SetKw{KwBreak}{break}
\SetKw{KwAnd}{and}
\SetKw{KwNot}{not}
\SetKwFunction{TestLeft}{ViolateCyc}
\SetKwProg{Fn}{Function}{}{}
\SetAlgoLined
 \KwIn{A tiling $f:G_{m,n}\to T$}
 \KwOut{A cyclically $\bar{L}$-dappled tiling $g: G_{m,n}\to T$}
 
\Begin{
 $g \gets f$ \;
 \For{$weight=0$ \KwTo $m+n-2$}{
	 \ForAll{$(i,j)$ such that $i+j=weight$}{
	 	\If{\TestLeft$(g,(i,j))$}{
		$g(i,j) \gets 1-g(i,j)$\;
		\If{\TestLeft$(g,(i,j))$}{
	 		$g(i,j) \gets g(i-1,j-1)$ \;
       		$g(i-1,j) \gets 1-g(i,j)$ \;
       		$g(i,j-1) \gets 1-g(i,j)$ \;
		}}
    }
 }
 \Return{$g$}
}
\Fn{\TestLeft$(f, (i,j))$}{
\ForAll{$\bar{H}^p_t \in \bar{L}$}{
\Switch{i}{
    \Case{$i=p-2$}{
	        \If{$f(0,j)=f(1,j)=\cdots=f(p-2,j)=t$}{\Return{true}}
    }
    \Case{$i=m-1$}{
        \If{there exists $k<p-2$ such that $f(k,j)=f(k-1,j)=\cdots=f(0,j)=
        f(m-1,j)=f(m-2,j)=\cdots=f(m-p+k,j)=t$}{\Return{true}}
    }
    \Case{$i>p-2$}{
    	\If{$f(i,j)=f(i-1,j)=\cdots=f(i-p,j)=t$}{\Return{true}}
    }
}
}
\ForAll{$\bar{V}^q_t \in \bar{L}$}{
\Switch{j}{
    Similar to the above.
}
}
\Return{false}
}
\caption{Algorithm to convert an input tiling to a cyclically $\bar{L}$-dappled one.}
\label{cyc-algorithm}
\end{algorithm}

\begin{example}\label{cyclic:fail}
One might wonder why we cannot just
replace (ii) above with 
\[
\text{(ii)' add $H^{p-1}_t$ to $L(i,j)$ if $i=p-1$}
\]
 to make it work when $p=2$.
In this case, we may have to add
$H^1_t$ to $L(m-1,j)$ in (iii),
which is problematic as we see 
in the following example with 
$\bar{L}=\{\bar{H}^3_0,\bar{V}^3_1\}$:
\[
\begin{array}{cccccc}
0 &0 &1 &0 &1 &0\\
1 &0 &1 &0 &1 &1\\
1 &1 &0 &1 &0 &1\\
0 &1 &1 &0 &0 &1\\
0 &0 &1 &1 &0 &\underline{1}\\
1& 0 &0 &0 &1 &1
\end{array}
\Leftrightarrow
\begin{array}{cccccc}
0 &0 &1 &0 &1 &0\\
1 &0 &1 &0 &1 &1\\
1 &1 &0 &1 &0 &1\\
0 &1 &1 &0 &0 &1\\
0 &0 &1 &1 &1 &0\\
1& 0 &0 &0 &\underline{1} &1
\end{array}
\]
Look at the tiling on the left.
Algorithm \ref{cyc-algorithm} with (ii) replaced by (ii)' does nothing up to the cell $(5,4)$ marked with $\underline{1}$.
Here we have $L(5,4)=\{ H^1_0, V^3_1\}$.
Rectifying the cell $(5,4)$ by Algorithm \ref{algorithm} will introduce
a new violation at $(4,5)$ as we see on the right,
and vice versa.
\end{example}

\begin{remark}\label{rem:cyclic_p=2}
If $\bar{L}$ consists of just two conditions 
$\{\bar{H}^p_0, \bar{V}^q_1\}$,
we can modify Algorithm \ref{cyc-algorithm}
further 
to make it work
even when $p=q=2$.
The idea is to make the first two rows and columns 
draughtboard.
Modify the input tiling to satisfy the following two conditions:
\begin{enumerate}[a)]
    \item $f(i,0)\neq f(i,1)$,
    $f(0,j)\neq f(1,j)$, $f(2k,0)\neq f(2k+1,0)$, and $f(0,2l)\neq f(0,2l+1)$
    \item $f(m-2,0)=f(0,n-2)$
\end{enumerate}
Then, the rest is rectified with Algorithm \ref{cyc-algorithm},
with (ii) replaced by
\[
\text{(ii)' add $H^{p-1}_t$ to $L(i,j)$ if $i=p-1$}.
\]
For the technical details, refer to
the implementation \cite{code}.
\end{remark}



\begin{example}
Fig. \ref{fig:cyclic} shows cyclically dappled tilings of $G_{10,6}$ obtained by 
Algorithm \ref{cyc-algorithm} for $\bar{L}=\{ \bar{H}^3_{\text{white}}, \bar{V}^3_{\text{orange}}\}$
and
by Remark \ref{rem:cyclic_p=2} for $\bar{L}=\{ \bar{H}^2_{\text{white}}, \bar{V}^2_{\text{orange}}\}$.
We repeated them twice both horizontally and vertically to obtain dappled tilings of $G_{20,12}$.
\begin{figure}[ht]
  \centering
  \begin{tabular}{ccc} 
      & Pattern & Repetition\\
      $p=q=3$ & \parbox[c]{0.3\linewidth}{
      \includegraphics[width=1\linewidth]{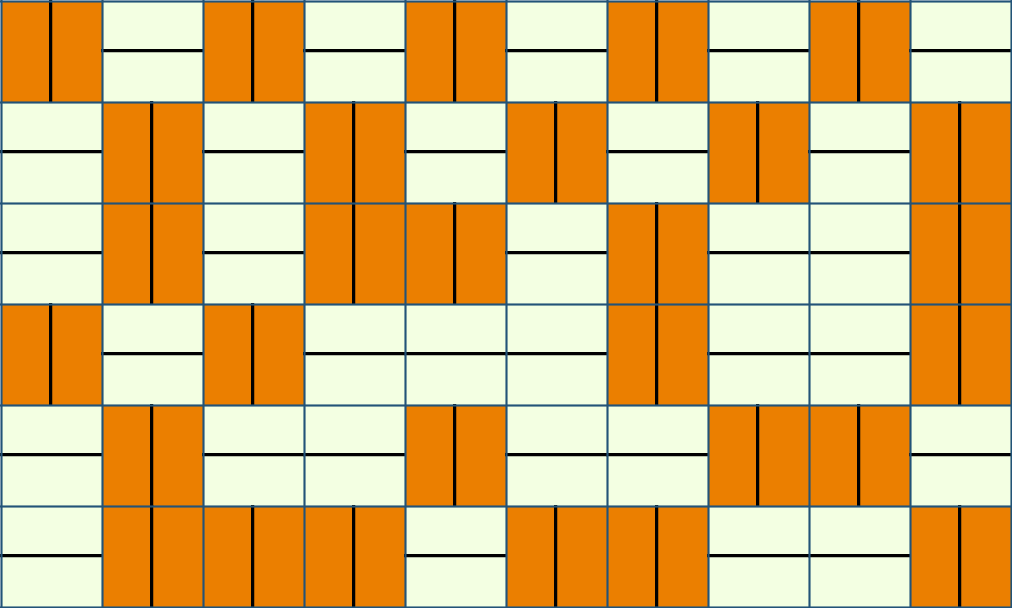}} & \parbox[c]{0.3\linewidth}{
      \includegraphics[width=1\linewidth]{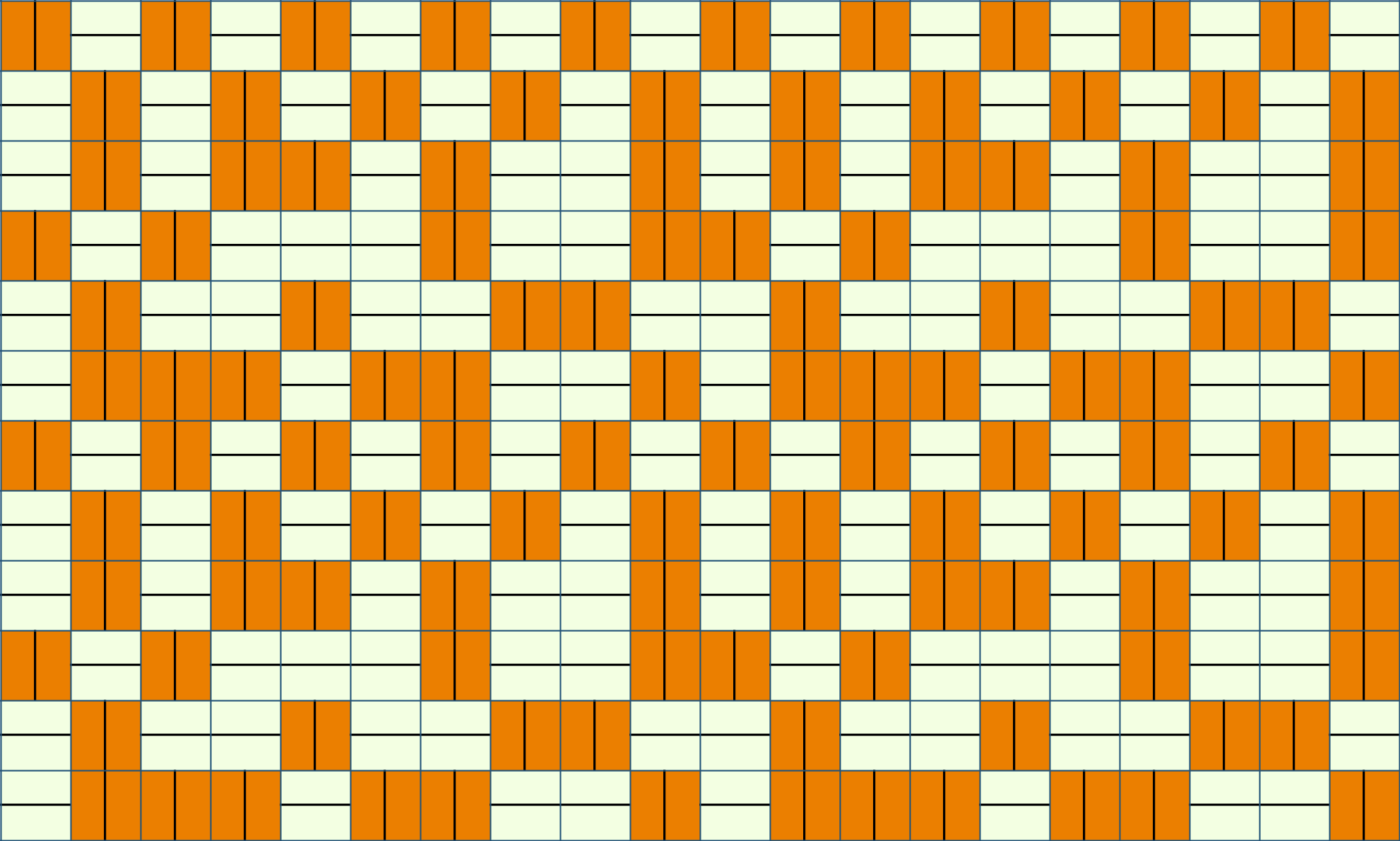}} \\
      & Pattern & Repetition\\
      $p=q=2$ & \parbox[c]{0.3\linewidth}{
      \includegraphics[width=1\linewidth]{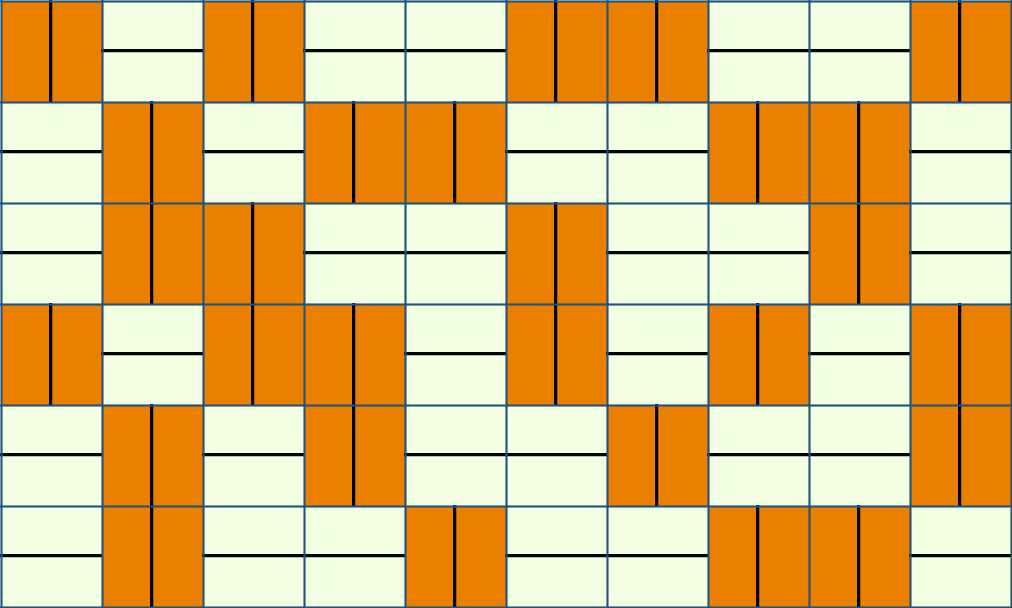}} & \parbox[c]{0.3\linewidth}{
      \includegraphics[width=1\linewidth]{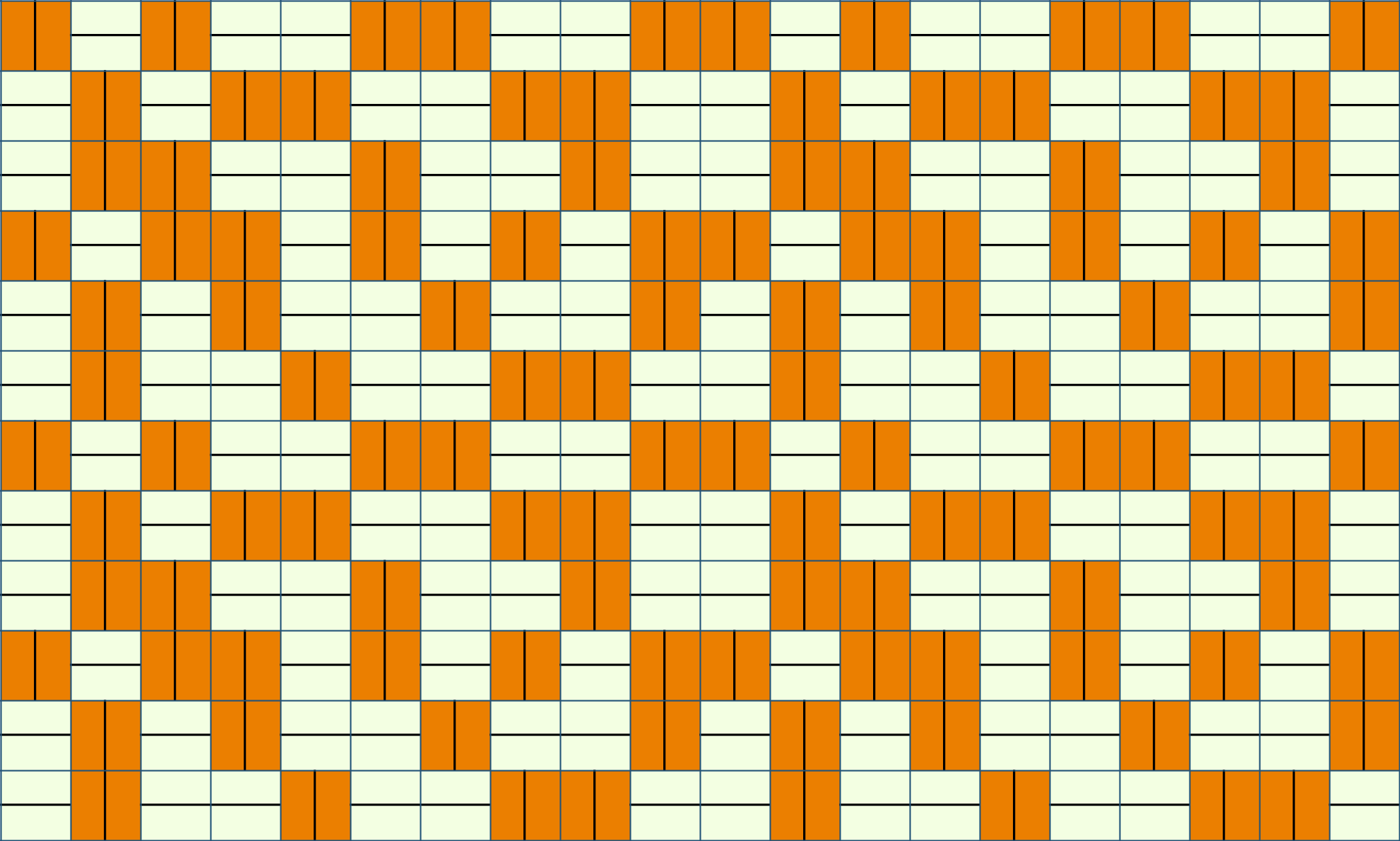}} \\
  \end{tabular}
  \caption{Cyclically dappled tilings obtained with our algorithm.} \label{fig:cyclic}
\end{figure}
\end{example}

\section{Example: Brick Wang Tiles}\label{sec:brick}
A method to create brick wall textures using the brick Wang tiles is introduced in
A.~Derouet-Jourdan et al.~\cite{meis2015} and studied further in \cite{AKM,scss2016}.
In this method, 
each tile represents how the corners of four bricks meet. 
It is assumed that the edges of the bricks are axis aligned and that each tile is traversed with a straight line, either vertically or horizontally. 
For aesthetic reasons, crosses are forbidden, where all four bricks are aligned and the corresponding tile is traversed by two straight lines. 
Formally, the set of {\em brick Wang tiles} $W$ is defined by
\[
W=\{ (c_1,c_2,c_3,c_4)\in C^4 \mid (c_1=c_3 \text{ and } c_2\neq c_4) \text{ or } (c_1 \neq c_3 \text{ and } c_2=c_4)\},
\]
where $C$ is a finite set, which we can think of as the set of
``positions'' of the brick edges (see Fig.~\ref{fig:brickwangtile}).
A tiling $\tau: G_{m,n}\to W$ is said to be a {\em valid Wang tiling} with $W$
if at all cells the positions of edges are consistent with those of the adjacent cells:
\begin{equation}\label{brick-wang}
 \tau(i,j)_1=\tau(i-1,j)_3, \tau(i,j)_2=\tau(i,j-1)_4 \quad
 (1\le i < m-1, 1\le j < n-1).
\end{equation}
Here, we do not pose any condition on the positions on the boundary;
we are concerned with the {\em free boundary} problem.
\begin{figure}[htb]
\centering
\resizebox{0.7\width}{0.7\height}{
\begin{tikzpicture}
[
    point/.style = {draw, circle,  fill = black, inner sep = 1pt},
]

\newcommand{\pythagwidth}{3cm}
\newcommand{\pythagheight}{2cm}
\newcommand{\tilesize}{3cm}

\coordinate [] (A) at (0, 0);
\coordinate [] (B) at (0, \tilesize);
\coordinate [] (C) at (\tilesize, \tilesize);
\coordinate [] (D) at (\tilesize, 0);

\draw [dashed] (A) -- (B) -- (C) -- (D) -- (A);

\coordinate [] (E) at (0.75 * \tilesize, \tilesize);
\coordinate [] (F) at (\tilesize, 0.5 * \tilesize);
\coordinate [] (G) at (0.25 * \tilesize, 0);
\coordinate [] (H) at (0, 0.5 * \tilesize);

\coordinate [] (E') at (0.75 * \tilesize, 0.5 * \tilesize);
\coordinate [] (G') at (0.25 * \tilesize, 0.5 * \tilesize);

\draw [very thick] (H) -- (F);
\draw [very thick] (G) -- (G');
\draw [very thick] (E) -- (E');

\node at (H) [point] {};
\node at (F) [point] {};
\node at (E) [point] {};
\node at (G) [point] {};


\draw[-latex,thick](2.5,-0.7 )node[right]
        {Brick edge} to[out=180,in=-90] (0.5 * \tilesize, 0.45 * \tilesize);
		
\draw[-latex,thick](-0.5*\tilesize, 0.6* \tilesize)node[left]
        {edge position $c_1$} to[out=-90,in=180] (0, 0.5 * \tilesize);

\draw[-latex,thick](1.1*\tilesize,1.2* \tilesize)node[right]
        {edge position $c_2$} to[out=-180,in=90] (0.75 * \tilesize, 1.05 * \tilesize);

\draw[-latex,thick](1.2*\tilesize, 0.6* \tilesize)node[right]
        {edge position $c_3$} to[out=-90,in=0] (\tilesize, 0.5 * \tilesize);

\draw[-latex,thick](-0.5*\tilesize,-0.7)node[left]
        {edge position $c_4$} to[out=0,in=-90] (0.25 * \tilesize, 0);

\end{tikzpicture}
}
\includegraphics[width=0.3\linewidth]{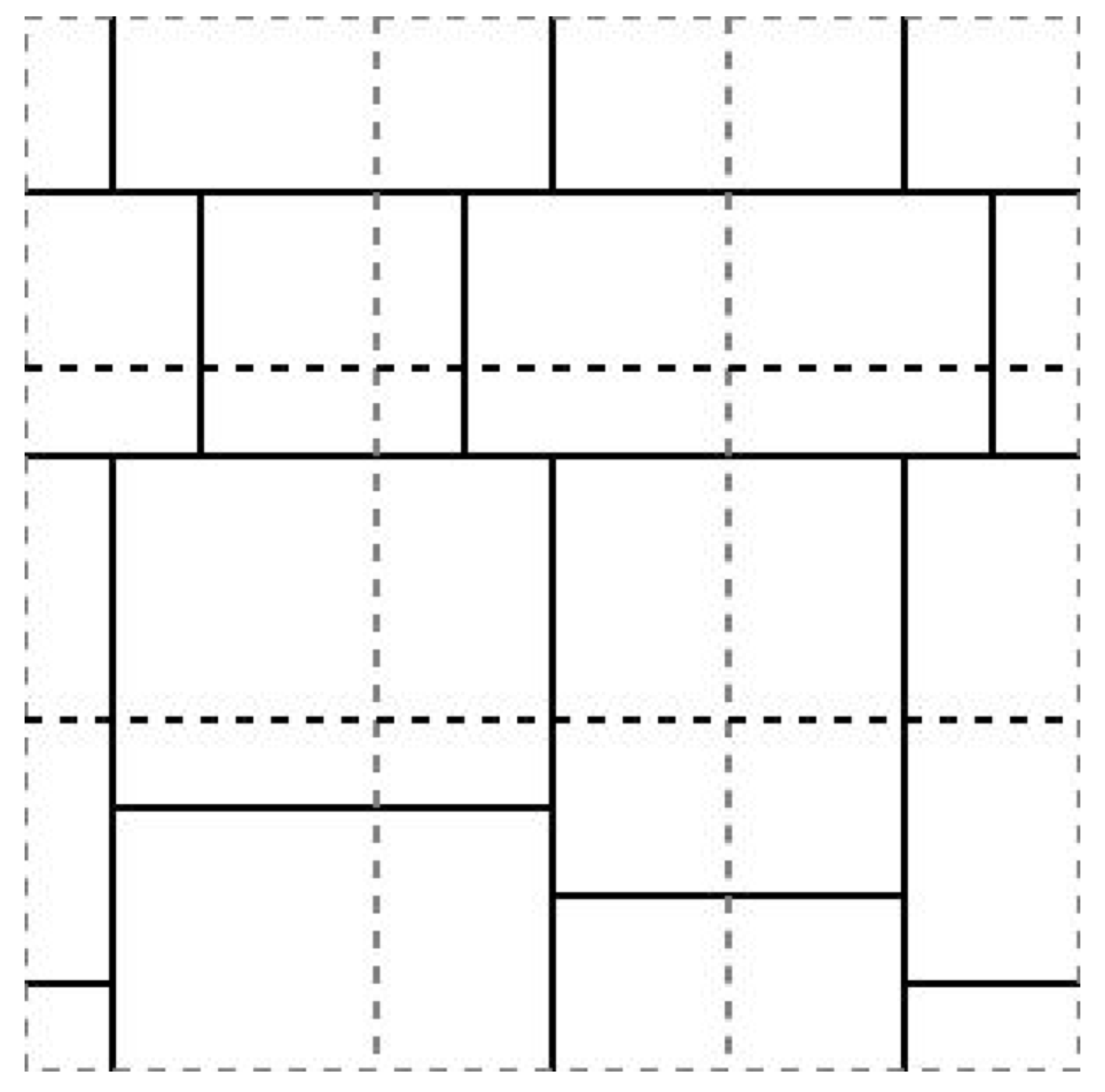}
\caption{A brick Wang tile and a $3\times 3$-tiling}\label{fig:brickwangtile}
\end{figure}
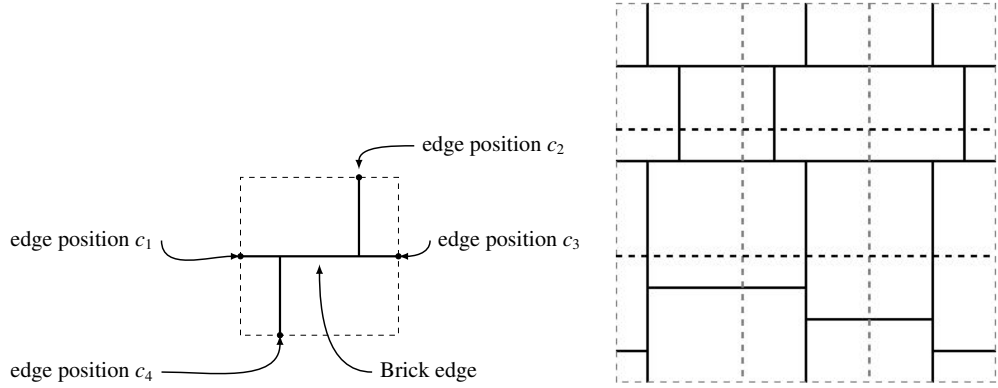

In \cite{AKM}, an algorithm to give a valid Wang tiling with $W$
for any planar region which contains a cycle
 and with {\em any} boundary condition.
 In this paper, we give a different approach to give
a brick pattern for a rectangular region in the plane
using our dappled tiling algorithm.
We restrict ourselves to the case of 
the free boundary condition and rectangular region,
but with the current approach we have a better control over the output.

A problem of the previous algorithms in \cite{meis2015,AKM} is
that it sometimes produces long traversal edges; 
horizontally consecutive occurrence of tiles with $c_1=c_3$
or vertically consecutive occurrence of tiles with $c_2=c_4$.
These are visually undesirable (see Fig. \ref{fig:comp-wall}).
\begin{figure}[ht]
  \centering
      \includegraphics[width=0.3\linewidth]{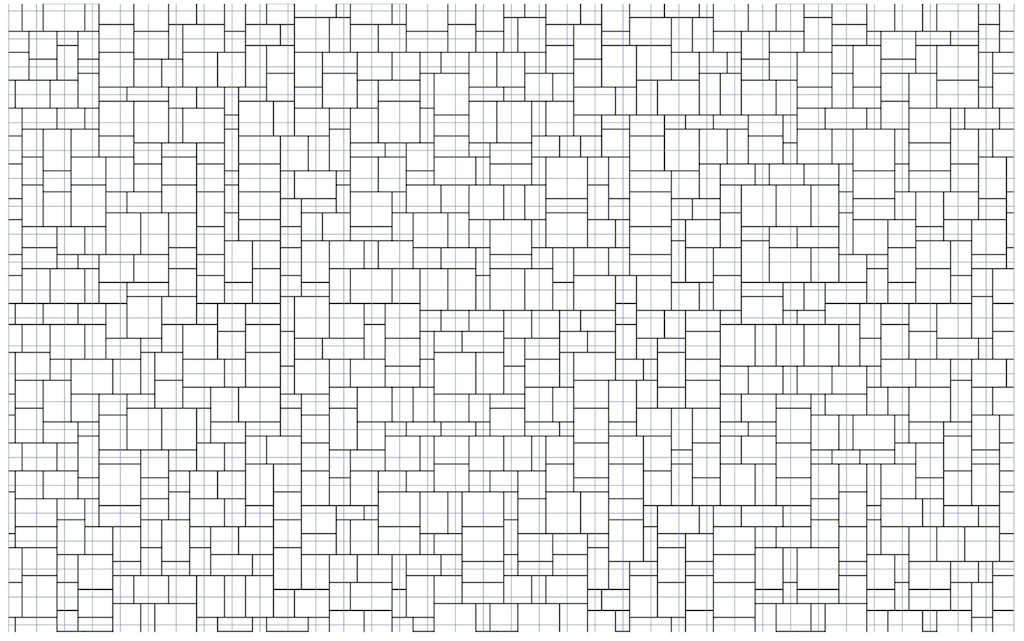} 
      \includegraphics[width=0.3\linewidth]{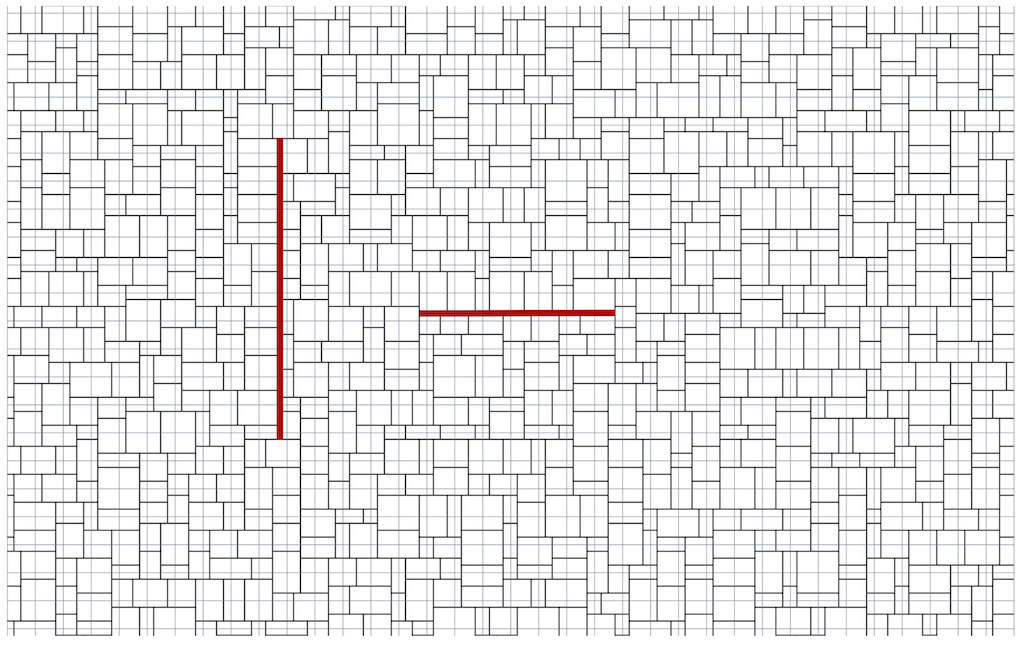}
      \includegraphics[width=0.3\linewidth]{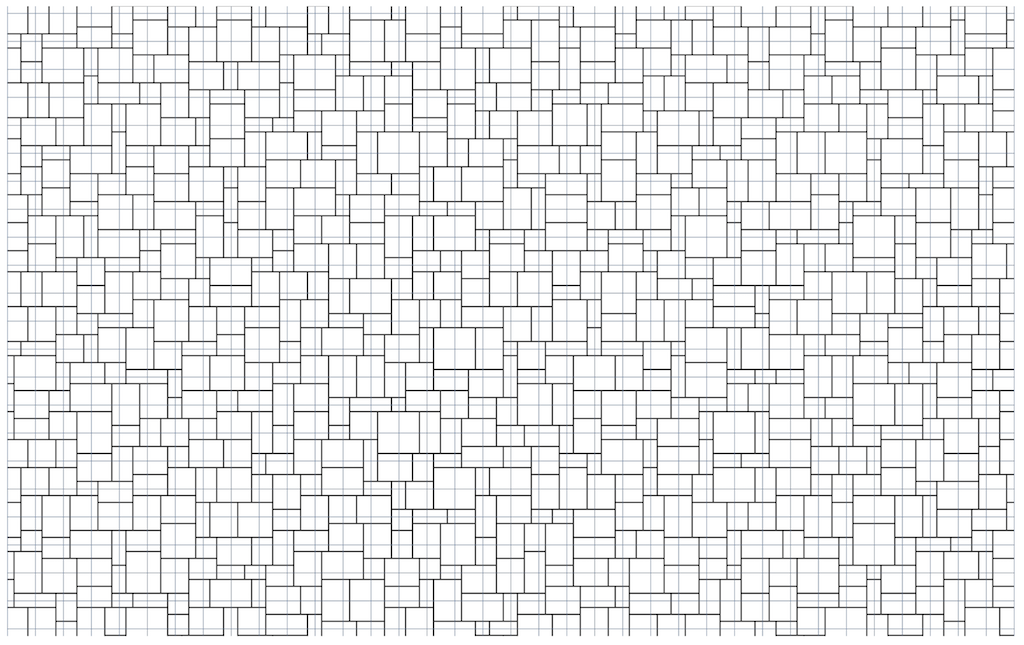}
  \caption{Brick Wall Patterns.Left: produced by the algorithms in \cite{meis2015,AKM},
  Middle: same as Left with emphasis on long traversal lines, Right: produced by our new algorithm  } \label{fig:comp-wall}
\end{figure}
We overcome this problem by our $L$-dappled tiling algorithm.
First, we divide $W$ into two classes
$W=W_0 \sqcup W_1$, where $W_0$ consists of those with $c_1=c_3$ and $W_1$ consists of those with $c_2=c_4$.
We label tiles in $W_0$ with $0$ and those in $W_1$ with $1$.
We now consider $L$-dappled tilings with $T=\{0,1\}$ and $L=\{H^p_0, V^q_1\}$,
which avoid horizontal strips bigger than $p$ consisting of tiles from $W_H$
and vertical strips bigger than $q$ consisting of tiles from $W_V$.

From an $L$-dappled tiling $f$, we can construct a valid Wang tiling with $W$:
Visit cells from left to right, and top to bottom.
At $(i,j)$, use \eqref{brick-wang} to determine edge positions for $c_1$ and $c_2$ (when $i,j>0$).
If $f(i,j)=0$, set $c_3=c_1$. Otherwise, set $c_4=c_2$.
Pick any positions randomly for the rest of the edges.
Obviously, this gives a valid Wang tiling with the desired property.

\begin{example}
Fig.~\ref{fig:wang} shows brick patterns constructed from tilings of $G_{10,6}$ with $T=\{0,1\}$.
The upper pattern, which is constructed from a user specified tiling,
 shows a clear diagonal pattern.
The lower pattern, which is constructed from the $L$-dappled tiling with $L=\{ H^2_{\text{white}}, V^2_{\text{orange}}\}$
produced by Algorithm \ref{algorithm} applied to the use specified tiling,
 looks more random while maintaining a subtle feel of the diagonal pattern.
\begin{figure}[ht]
  \centering
  \begin{tabular}{ccc} 
      & initial tiling & corresponding Brick Wang tiles\\
       & \parbox[c]{0.3\linewidth}{
      \includegraphics[width=1\linewidth]{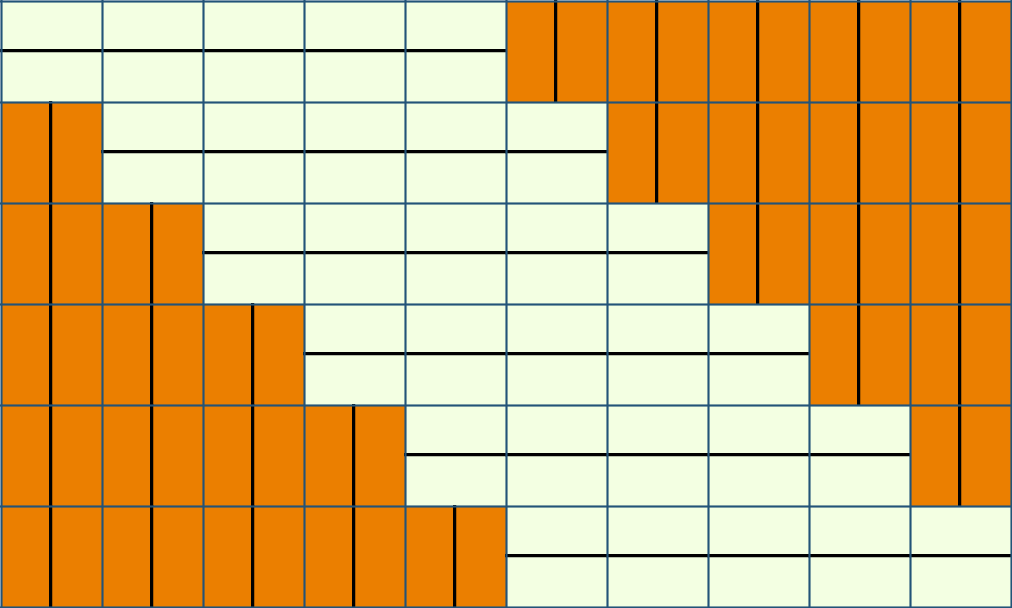}} & \parbox[c]{0.3\linewidth}{
      \includegraphics[width=1\linewidth]{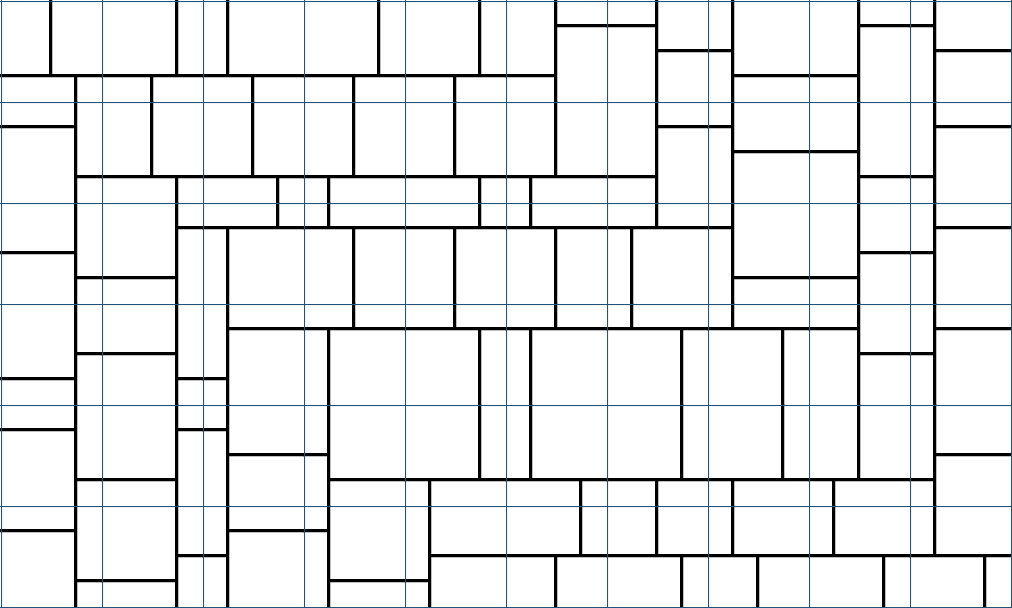}} \\
      & dappled tiling & corresponding Brick Wang tiles\\
       & \parbox[c]{0.3\linewidth}{
      \includegraphics[width=1\linewidth]{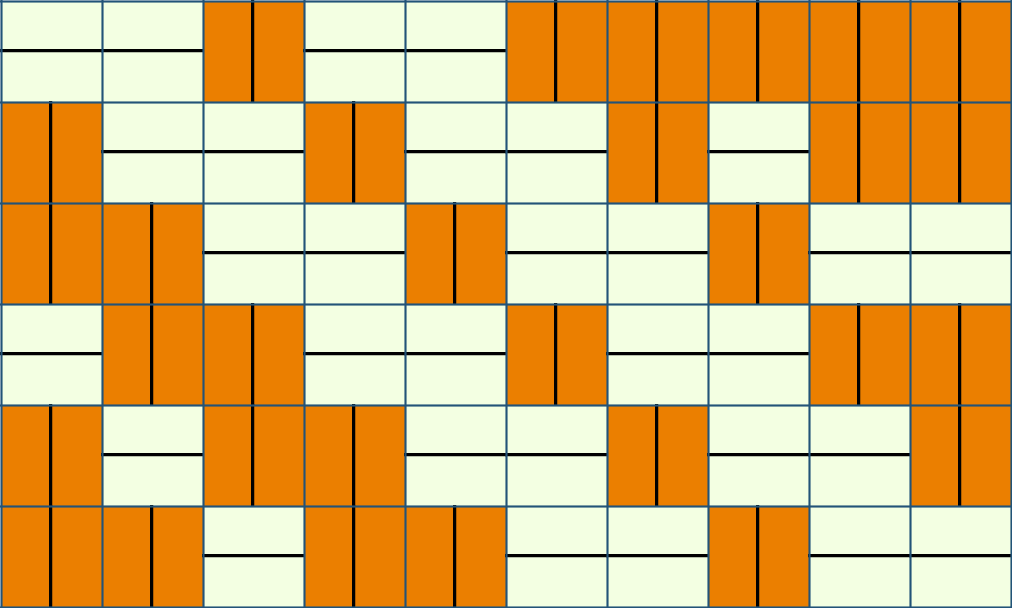}} & \parbox[c]{0.3\linewidth}{
      \includegraphics[width=1\linewidth]{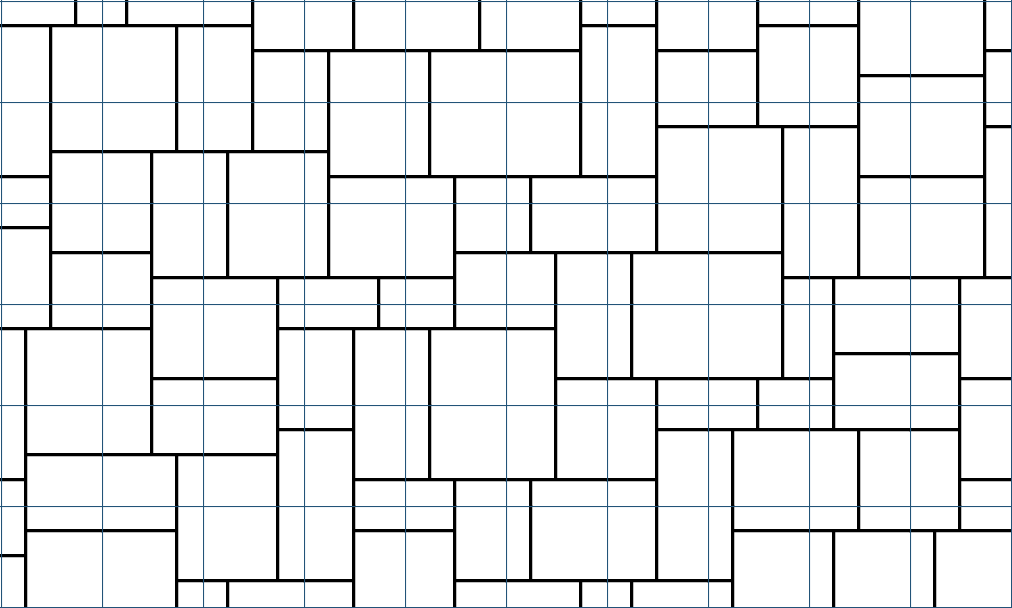}} \\
  \end{tabular}
  \caption{Dappled tiling and corresponding Brick Wang tiling} \label{fig:wang}
\end{figure}

\end{example}

\section{Example: Flow Tiles}\label{sec:flow}
Consider an $L$-dappled tiling with
$T=\{-, |\}$ and $L=\{H^{p_-}_{-}, H^{p_|}_{|}, V^{q_{-}}_{-}, V^{q_|}_{|}\}$.
We interpret it as a flow field to create a crowd simulation.
We start with particles spread over the tiling. They move around following
the ``guidance'' specified by the tile.
More precisely, each particle picks a direction according to the tile on which it locates.
For example, assume a particle is at a cell with $-$.
Then, choose either left or right and move in the direction.
When it reaches the centre of an adjacent tile, say with $|$, 
chooses either up or down and continues in the direction.
See the supplementary video \cite{video}.

\section{Conclusion and Future work}
We defined the notion of dappled tilings, 
which is useful to produce texture 
patterns free of a certain kind of repetition.
We gave an efficient algorithm (Algorithm \ref{algorithm}) to convert any tilings to a dappled one.
Our method has the following advantages.
\begin{itemize}
\item It produces all the dappled tilings if we start with a random tiling.
This is because the algorithm does not modify the input tiling if it is already $L$-dappled.
\item It has some control over the distribution of tiles since we can specify the initial tiling.
\end{itemize}
We also discussed an algorithm (Algorithm \ref{cyc-algorithm}) 
to convert any tilings to a cyclically dappled one. 
Cyclically dappled tilings can be used repeatedly to fill a larger region.
However, Algorithm \ref{cyc-algorithm} is limited in the sense that it does not produce
all the possible cyclically dappled tilings.

We finish our discussion with a list of future work
which encompasses both the theoretical and the practical problems.
\begin{enumerate}[{[}1{]}]
\item The number of $L$-dappled tilings of $G_{m,n}$ with a given set $L$ of conditions: 
to determine an explicit or recursive formula is a mathematically interesting problem.
\item A better cyclic algorithm: in \S \ref{sec:cyclic} we gave an algorithm
to produce cyclically dappled tilings with some limitations.
It would be good to get rid of these limitations.
\item Conditions specified by subsets:
For $\tau \subset T$, we define the condition $H^p_\tau$ which prohibits
horizontal strips consisting of $p+1$ tiles in $\tau$.
We would like to give an algorithm to produce $L$-dappled tilings, 
where $L$ consists of this kind of generalised conditions.
For example, by setting $L=\{H^2_{\{\text{white, grey}\}}, V^2_{\{\text{grey, black}\}} \}$ we can produce tilings without long strips of similar colour.
\item Closest dappled tiling:
Our algorithm takes a tiling as input and produces an $L$-dappled tiling, 
which is usually not very different from the input. 
However, the output is not the closest solution in terms of the Hamming distance
$d(f_1,f_2)=|\{(i,j)\in G_{m,n}\mid f_1(i,j)\neq f_2(i,j)\}|$.
\begin{example}
For $L=\{H^2_0, V^2_1\}$, Algorithm \ref{algorithm} converts
\[
\begin{matrix}
0 & 0 & 0  \\
1 & 0 & 1  \\
0 & 0 & 1 
\end{matrix}
\Rightarrow
\begin{matrix}
0 & 0 & 1  \\
1 & 0 & 1  \\
0 & 1 & 0 
\end{matrix}
\]
but one of the closest dappled tilings to the input is
\[
\begin{matrix}
0 & 1 & 0  \\
1 & 0 & 1  \\
0 & 0 & 1 
\end{matrix}
\]
\end{example}
It is interesting to find an algorithm to produce an $L$-dappled tiling
closest to the given tiling.
\item Extension of the flow tiling in \S \ref{sec:flow}:
we can consider different kinds of tiles such as
emitting/killing tiles, where new particles are born/killed,
and speed control tiles, where the speed of a particle is changed.
\item A parallel algorithm: our algorithm is sequential but it is desirable to have a parallelised algorithm.
We may use a cellular automaton approach.
\item Global constraints: the conditions we consider in the $L$-dappled tiling is {\em local} in the sense that
they can be checked by looking at a neighbourhood of each cell.
Global constraints such as specifying the total number of a particular tile
can be useful in some applications. 
We would like to generalise our framework so that we can deal with global constraints.
\item Boundary condition:
given a partial tiling of $G_{m,n}$, we can ask to extend it to an $L$-dappled tiling.
A typical example is the case where the tiles at the boundary are specified.
In the cyclic setting, it is not even trivial to determine if there is a solution or not.
\begin{example}
Consider a $4\times 4$-grid with $\bar{L}=\{\bar{H}^2_0, \bar{V}^2_1\}, T=\{0,1\}$ and the following partial tiling:
\[
\begin{matrix}
1 & ? & ? & ?\\
? & 0 & ? & ?\\
? & 0 & 1 & 1 \\
? & 0 & ? & ?
\end{matrix}
\]
There exists no cyclically $\bar{L}$-dappled tiling extending 
(obtained by filling the cells marked with ``$?$'')
the given one. 
This is because in a $4\times4$ cyclically $\bar{L}$-dappled tiling,
there should be an equal number of $0$ and $1$.
This implies there should be exactly two $1$'s in each column,
which is not the case with the above example.

For a larger board $G_{m,n}$, where $m\ge 7$, $n \ge 4$, and $m-1$ is divisible by $3$,
we have a similar example:
\[
\begin{matrix}
\cdots ? & 0 & ? & ? & 0 & ? \cdots \\
\cdots ? & 0 & ? & ? & 0 & ? \cdots \\
\cdots ? & 0 & 1 & 1 & 0 & ? \cdots \\
\cdots ? & 0 & ? & ? & 0 & ? \cdots \\
\cdots ? & 0 & ? & ? & 0 & ? \cdots \\
\cdots ? & 0 & 1 & 1 & 0 & ? \cdots \\
\cdots ? & 0 & ? & ? & 0 & ? \cdots \\
\cdots ? & 0 & ? & ? & 0 & ? \cdots \\
\cdots ? & 0 & 1 & 1 & 0 & ? \cdots \\
\cdots ? & 0 & ? & ? & 0 & ? \cdots \\
\cdots ? & 0 & ? & ? & 0 & ? \cdots \\
\cdots ? & 0 & 1 & 1 & 0 & ? \cdots \\
\cdots ? & 0 & ? & ? & 0 & ? \cdots 
\end{matrix}
\]
There exists no cyclically $L$-dappled tiling extending it.
This can be checked, for example, by choosing a tile for $(3,3)$ and
continue filling cells which are forced to have either $0$ or $1$ by the conditions.
No matter what tile we choose for $(3,3)$, we encounter violation at some point.

We would like to have a more efficient algorithm to decide and solve tiling problems with boundary conditions.
\end{example}
\item Interpretation as a SAT problem: the $L$-dappled tiling is 
a satisfiability problem and it would be interesting to formalise it
to give a formal verification of the algorithm.
\end{enumerate}

\section*{Acknowledgements}
A part of this work was conducted during
the IMI Short Term Research project ``Formalisation of Wang tiles for texture synthesis'' at Kyushu University.
The authors thank Kyushu University for the support.
The authors are grateful to Yoshihiro Mizoguchi for his helpful comments.

\end{document}